\documentclass{statsimp}
\begin{document}
\title{An Improved Solution to the Two Normal Means Problem via Regularization}
\author{Yang Liu\footnote{Department of Human Development and Quantitative Methodology, University of Maryland, College Park. Correspondence author. Email: yliu87@umd.edu}
  \qquad Jonathan P. Williams\footnote{Department of Statistics, North Carolina State University}
}

\abstract{%
  The many-normal-means problem is a classic example that motivates the development of many important inferential procedures in the history of statistics. In this short note, we consider a further special case of the problem, which involves only two normally distributed data points with a constraint that the pair of means  are not too far apart from one another. Starting with a regularized ML estimator, we construct a novel possibilistic IM for marginal inference on one of the two means. Not only does the new IM remain valid, it is also more efficient than the standard marginal inference ignoring the {\it a priori} information about the closeness of means, as well as the partial conditioning IM solution recently proposed in \citet{YangEtAl2023}.\\\bigskip
}

\maketitle

\section{Introduction}
\label{s:intro}

The problem of many-normal-means has been treated as a benchmark example for statistical inference and has been extensively studied in the literature \cite[e.g.,][]{Liu2022, Stein1956, YangEtAl2023}. In the most classic form, the problem concerns a sequence of independent but not identically normal random variables such that $Y_i\sim\sN(\theta_i, 1)$, $i = 1,\dots, n$. The parameters of interests are the mean sequence $\{\theta_i\}$. Recently, \citet{YangEtAl2023} investigated an interesting special case of the problem when the consecutive means in the sequence are in the neighborhood of one another, referred to as the H\"older constraints. In its simplest form, the constraints map onto the following restricted parameter space:
\begin{equation}
  \{\theta = (\theta_1, \dots, \theta_n): |\theta_{i + 1} - \theta_i|\le B,\ i = 1,\dots, n - 1\},
  \label{eq:holder}
\end{equation}
in which $\theta$ is the collection of all the mean parameters, and $B > 0$ is a known constant.  An inferential model (IM) solution is derived in \citet{YangEtAl2023} with a partial conditioning argument, which interpolates between the existing conditional and marginal IM solutions \citep{MartinLiu2015a, MartinLiu2015b}.  It is demonstrated both analytically and numerically that, compared to competitive methods, the partially conditional IM yields shorter confidence intervals (CIs) for a focal mean parameter and the CIs are valid in the sense that the achieve their nominally stated frequentist coverage for finite samples.

Motivated by the derivation of \citet{YangEtAl2023}, we develop an alternative strategy to improve the efficiency of valid inferential procedures. The proposed approach hinges upon regularization, a technique of central importance in contemporary statistics. For ease of illustration, we focus on a simple two-means problem (i.e., $n = 2$) that was discussed in Sections 2 and 3 of \citet{YangEtAl2023}. In this special case, the H\"older constraint (\ref{eq:holder}) reduces to requiring $|\theta_2 - \theta_1|\le B$. We show that \citeauthor{YangEtAl2023}'s (\citeyear{YangEtAl2023}) partial conditioning solution can be equivalently obtained from a regularized maximum likelihood (ML) estimator of the means. Moreover, we identify an even more efficient IM solution based on the same regularized estimator for the special case of $n = 2$.

Our paper has three main contributions. First, we provide a modernized and more prescriptive argument to reproduce the partial conditioning solution of \citet{YangEtAl2023}, using possiblistic IM, instead of the random-set IM construction employed in \citet{YangEtAl2023}. Second, we show that a more efficient solution naturally emerges from the possibilistic IM construction, at least for the two-normal-means sub-problem. Third, we demonstrate how to incorporate regularization in the possibilistic IM construction, which not only showcases the flexibility of the IM but also highlights the key role of regularization in obtaining efficient inference.

The rest of the paper is organized as follows. We begin with a brief overview of the possiblistic IM framework. We then reconstruct the \cite{YangEtAl2023} partial conditioning IM solution with an alternative argument based a Wald-type test statistic computed from the regularized ML estimator. Next, we introduce our new IM solution that leads to more efficient inference, utilizing the same regularized ML estimator but a different test statistic. We demonstrate both analytically and numerically that our solution is more efficient than that of \cite{YangEtAl2023}. The paper is concluded with discussions of limitations and possible extensions.

\section{Foundations of Inferential Models}
\label{s:im}
IM is a completely general inferential framework that allows valid probabilistic inference with or without prior information. To situate our discussion, we present a concise overview of IM based on the calculus of possibility measures \cite[e.g.,][]{DuboisPrade1988, Dubois2006}. More details about the theory of possibilistic IM can be found in \citet{LiuMartin2024}, \citet{Martin2022a, Martin2022b, Martin2022c}, and \citet{Martin2025}.

Suppose that the data $Y \in\cY$ follow the distribution $\pr_{Y|\theta^*}$, in which $\theta^*\in\cQ$ denotes the data generating (i.e., true) parameters, $\cY$ is the data space, and $\cQ$ is the parameter space. Given observed data $y\in\cY$, an IM represents degrees of belief about each parameter value in $\cQ$ through a data-dependent map onto the unit interval $[0, 1]$. This map, denoted by $\pi_y:\cQ\to[0, 1]$, should satisfy $\sup_{\theta\in\cQ}\pi_y(\theta) = 1$ and thus is a \emph{possibility contour}.\footnote{We adopt this term from \citet{LiuMartin2024}. In the literature of possibility theory \citep[e.g.,][]{DuboisPrade1988}, such a function is more commonly referred to as a ``possibility distribution.''} An IM possibility contour must also be \emph{valid} in the frequentist sense: For any hypothesis $H\subseteq\cQ$ and $\alpha\in[0, 1]$,
\begin{equation}
  \pr_{Y|\theta}\{\pi_Y(\theta)\le\alpha\} \le\alpha.
  \label{eq:valid}
\end{equation}
An IM possibility contour can be conveniently constructed as a $p$-value function using a test statistic $T: \cY\times\cQ\to\real$ by
\begin{equation}
  \pi_y(\theta) = \pr_{Y|\theta}\{T(Y, \theta)\ge T(y, \theta)\}.
  \label{eq:contour}
\end{equation}
The validity requirement (\ref{eq:valid}) is established for (\ref{eq:contour}) from the probability integral transform \citep[e.g.,][Theorem 2.1.10 and Exercise 2.10]{CasellaBerger2002}. The test statistic $T$ in (\ref{eq:contour}) can be arbitrarily chosen. Standard constructions of possibilistic IM are often based on the relative likelihood ratio statistic \citet{Martin2022a, Martin2022b, Martin2022c}. By establishing a possibilistic Bernstein-von Mises theorem, it is shown in \citet{MartinWilliams2025} that using the relative likelihood ratio statistic leads to asymptotically efficient inference.  The asymptotic efficiency, however, relies on the classical Cram{\'e}r or Le Cam regularity conditions for consistency of an ML estimate; in particular, it assumes a fixed parameter dimension as the sample size is taken to infinity.  Accordingly, the efficiency of the relative likelihood-based possibilistic IM construction does not apply to the over-parametrized, many-normal-means problem.  To deal with the over-parameterization, additional structure is needed, and a popular approach is to incorporate regularization constraints.  We therefore consider and develop {\em regularized}, relative-likelihood-ratio-based IM extensions.


In many practical problems, we are only interested in making inference about a single coordinate of $\theta$, referred to as the \emph{focal parameter}. Without loss of generality, partition $\theta = (\varphi, \nu)$, in which $\varphi\in\real$ denotes the single focal parameter and $\nu$ denotes the \emph{nuisance parameters}. Given a possibility contour $\pi_y$ for overall inference on $\theta$ that satisfies (\ref{eq:valid}), a marginal possibility contour for $\varphi$ can be obtained by taking the supremum over the nuisance parameters:
\begin{equation}
  \varpi_y(\varphi) = \sup_{\nu: (\varphi, \nu)\in\cQ}\pi_y(\varphi, \nu),
  \label{eq:mcontour}
\end{equation}
in which $\pi_y(\varphi, \nu)$ is a shorthand for $\pi_y\{(\varphi, \nu)\}$. In fact, (\ref{eq:mcontour}) is valid for marginal inference on $\varphi$ in that
\begin{equation}
  \pr_{Y|\varphi, \nu}\{\varpi_Y(\varphi)\le\alpha\}\le\pr_{Y|\varphi, \nu}\left\{\pi_y(\varphi, \nu)\le\alpha\right\}\le\alpha.
  \label{eq:validm}
\end{equation}
Marginal IM based on (\ref{eq:mcontour}) can be used to generate general purpose inference for the focal parameter $\varphi$, but the present work is primarily concerned with the construction of CIs. It follows from (\ref{eq:validm}) that the \emph{upper $\alpha$-cut} of $\varpi_y$, 
\begin{equation}
  C(\alpha; y) = \{\varphi: \varpi_y(\varphi) > \alpha\},
  \label{eq:mci}
\end{equation}
is a $100(1-\alpha)\%$ confidence region for $\varphi$. While (\ref{eq:mci}) is not necessarily an interval, though this is the case in all subsequent examples.

As an illustration, we apply the IM framework to derive the standard CI for the two-normal-means problem. Recall that $Y_1\sim\sN(\theta_1, 1)$ and $Y_2\sim\sN(\theta_2, 1)$ are two independent normal random variables. Here, $\theta = (\theta_1, \theta_2)$ and we focus on the marginal inference for $\theta_2$ (i.e., $\varphi = \theta_2$ and $\nu = \theta_1$). Without having knowledge about any relationship between $\theta_1$ and $\theta_2$ or $Y_{1}$ and $Y_{2}$, the common practice in making marginal inference on $\theta_2$ is to completely ignore the observed $y_1$.  Consider the canonical test statistic $T(Y, \theta) = (Y_2 - \theta_2)^2$, which follows a $\chi^2(1, 0)$ distribution (i.e., a chi-square distribution with degree of freedom 1 and noncentrality parameter 0) under $\pr_{Y|\theta}$. Then the marginal possibility contour based on ignorance can be expressed as
\begin{equation}\label{eq:simpcontour}
\begin{aligned}
\varpi_{y_2}(\theta_2) & = \sup_{\theta_1: (\theta_1,\theta_2)\in\cQ}\pi_y(\theta) = \sup_{\theta_1: (\theta_1,\theta_2)\in\cQ}\pr_{Y|\theta}\{(Y_2 - \theta_2)^2\ge (y_2 - \theta_2)^2\} \\
& = \pr_{Y_2|\theta_2}\{(Y_2 - \theta_2)^2\ge (y_2 - \theta_2)^2\} = 1 - F\left\{ (y_2 - \theta_2)^2; 1, 0 \right\},
\end{aligned}
\end{equation}
in which $F(\cdot; k, \gamma)$ denotes the cumulative distribution function (CDF) of the $\chi^2(k, \gamma)$ distribution (with degrees of freedom $k > 0$ and noncentrality parameter $\gamma\ge 0$). The upper $\alpha$-cut of (\ref{eq:simpcontour}) yields the standard textbook CI for $\theta_2$:
\begin{equation}
  C_0(\alpha; y) = [y_2 - z_{1-\alpha/2}, y_2 + z_{1 - \alpha/2}],
  \label{eq:stdci}
\end{equation}
in which $z_\beta$ denotes the $\beta$th quantile of the standard Gaussian distribution. The confidence limits in (\ref{eq:stdci}) are obtained by solving for $\theta_2$ from $\varpi_{y_2}(\theta_2) = \alpha$, noting that $z_{1-\alpha/2}^2$ coincides with the $(1-\alpha)$th quantile of $\chi^2(1, 0)$.

\section{An Alternative Construction of the Partial Conditioning Solution}
\label{s:first}
Leveraging the H\"older constraint, $|\theta_2 - \theta_1|\le B$, it is shown in \citet{YangEtAl2023} that the standard CI (\ref{eq:stdci}) can be improved. The intuition is that if $\theta_1$ is known to be in the vicinity of $\theta_2$, then observing $y_1$ in addition to $y_2$ should provide more information about $\theta_2$ than observing $y_2$ alone. The derivation of \citet{YangEtAl2023} invokes a partial conditioning argument in combination with the classic IM formulation based on predictive random sets. In this section, we show that the same solution can be obtained via a more straightforward construction of possibilistic IM using a regularized estimator of $\theta_2$.

Consider the regularized negative log-likelihood function for the two-normal-means problem
\begin{equation}
  \ell(\theta, \lambda; y) = \frac{(y_1 - \theta_1)^2}{2} + \frac{(y_2 - \theta_2)^2}{2} + \frac{\lambda(\theta_1 - \theta_2)^2}{2},
  \label{eq:rnll}
\end{equation}
in which $\lambda$ is a non-negative penalty weight, and constants irrelevant to $\theta$ and $y$ are omitted. In (\ref{eq:rnll}), we incorporate a ridge-type penalty on the difference $\theta_1 - \theta_2$, a choice motivated by differentiability and the {\it a priori} information that the magnitude of the difference is small. Because (\ref{eq:rnll}) is convex and quadratic in $\theta$, it has a unique minimizer
\begin{equation}
  \hat\theta(y; \lambda) = \left[\hat\theta_1(y; \lambda),
  \hat\theta_2(y; \lambda)\right]= \left[
  \frac{(1 + \lambda)y_1 + \lambda y_2}{1 + 2\lambda},
\frac{\lambda y_1 + (1 + \lambda)y_2}{1 + 2\lambda}\right],
  \label{eq:rmle}
\end{equation}
which is referred to as the \emph{regularized ML estimator} of $\hat\theta$. When $\lambda = 0$, the regularized estimator reduces to the usual ML estimator, $\hat\theta(y; 0) = y$.
 
To perform inference on $\theta_2$, we construct a test statistic based on the second coordinate of (\ref{eq:rmle}), $\hat\theta_2(y; \lambda)$. Note that
\begin{equation}
  (1 + 2\lambda)\left[\hat\theta_2(Y; \lambda) - \theta_2\right] = \lambda (Y_1 - \theta_2) + (1 + \lambda)(Y_2 - \theta_2)\sim\sN(\lambda(\theta_1 - \theta_2), \lambda^2 + (1 + \lambda)^2)
  \label{eq:diff}
\end{equation}
under $\pr_{Y|\theta}$; therefore, we define the following central chi-square statistic:
\begin{equation}
  T_1(Y, \theta; \lambda) =  \frac{\left[\lambda (Y_1 - \theta_2) + (1 + \lambda)(Y_2 - \theta_2) - \lambda(\theta_1 - \theta_2)\right]^2}{\lambda^2 + (1 + \lambda)^2}.
  \label{eq:wald1}
\end{equation}
The joint possibility contour corresponding to (\ref{eq:wald1}) is
\begin{equation}
  \begin{aligned}
    \pi_{y, 1}(\theta; \lambda) &=  \pr_{Y|\theta}\{T_1(Y, \theta; \lambda)\ge T_1(y, \theta; \lambda)\}\\
&= 1 - F\left( \frac{\left[\lambda (Y_1 - \theta_2) + (1 + \lambda)(Y_2 - \theta_2) + \lambda\theta_2 - \lambda\theta_1\right]^2}{\lambda^2 + (1 + \lambda)^2} ; 1, 0\right).
  \end{aligned}
  \label{eq:waldcontour1}
\end{equation}
Under the H\"older constraint, the marginal contour for $\theta_2$ is obtained by taking the supremum of $\pi_{y, 1}(\theta; \lambda)$ with respect to $\theta_1$ over the interval $[\theta_2 - B, \theta_2 + B]$:
\begin{equation}
  \begin{aligned}
    &\varpi_{y, 1}(\theta_2; \lambda) = \sup_{\theta_1\in[\theta_2 - B, \theta_2 + B]}\pi_{y, 1}(\theta; \lambda)\\
    &=  \begin{cases}
      \displaystyle 1 - F\left( \frac{[\lambda(y_1 - \theta_2) + (1 + \lambda)(y_2 - \theta_2) - \lambda B]^2}{\lambda^2 + (1 + \lambda)^2}; 1, 0 \right),&\displaystyle\theta_2\le\frac{\lambda y_1 + (1 + \lambda)y_2 - \lambda B}{1 + 2\lambda};\\
    \displaystyle 1 - F\left( \frac{[\lambda(y_1 - \theta_2) + (1 + \lambda)(y_2 - \theta_2) + \lambda B]^2}{\lambda^2 + (1 + \lambda)^2}; 1 ,0\right),&\displaystyle\theta_2\ge\frac{\lambda y_1 + (1 + \lambda)y_2 + \lambda B}{1 + 2\lambda};\\
    1,&\hbox{otherwise.}
  \end{cases}
  \end{aligned}
  \label{eq:mwaldcontour1}
\end{equation}
To see why (\ref{eq:mwaldcontour1}) holds, note that the graph of (\ref{eq:waldcontour1}), when viewed as a function of $\theta_2$, reaches the maximum 1 at $\theta_2 = (1+\lambda)^{-1}[\lambda y_1 + (1 + \lambda)y_2 - \lambda\theta_1]$. As we vary $\theta_1$ within the interval $[\theta_2 - B, \theta_2 + B]$, the graph simply shifts with the mode moving between $(1+2\lambda)^{-1}[\lambda y_1 + (1 + \lambda)y_2 \pm \lambda B]$. For $\theta_2$ values within this interval, the supremum of (\ref{eq:waldcontour1}) is always 1. To the left (resp. right) of the interval, the supremum traces the version of (\ref{eq:waldcontour1}) when $\theta_1 = \theta_2 - B$ (resp. when $\theta_1 = \theta_2 + B$). 

By the connection between the $\chi^2(1, 0)$ and $\sN(0, 1)$ distributions, the $\alpha$-cut of (\ref{eq:mwaldcontour1}), which is a marginal CI for $\theta_2$, has the following explicit expression:
\begin{equation}
  \begin{aligned}
    C_1(\alpha; y) = \bigg[  &
    \frac{\lambda y_1  + (1 + \lambda)y_2 - \lambda B - z_{1-\alpha/2}\sqrt{\lambda^2  + (1 + \lambda)^2}}{1 + 2\lambda},\\
    &\frac{\lambda y_1 + (1 + \lambda)y_2 + \lambda B + z_{1-\alpha/2}\sqrt{\lambda^2  + (1 + \lambda)^2}}{1 + 2\lambda}
  \bigg].
  \end{aligned}
  \label{eq:cut1}
\end{equation}
The length of (\ref{eq:cut1}) is given by
\begin{equation}
  L_1(\lambda; \alpha, B) = \frac{2}{1 + 2\lambda}\left[ \lambda B + z_{1-\alpha/2}\sqrt{\lambda^2 + (1 + \lambda)^2} \right],
  \label{eq:len1}
\end{equation}
which is not dependent on data $y$ or parameters $\theta$. For fixed $\alpha\in(0, 1)$ and $B\ge 0$, we can find the optimal penalty weight as the minimizer of (\ref{eq:len1}):
\begin{equation}
  \lambda_1^*(\alpha, B) = \frac{-B + \sqrt{-B^2 + 2z_{1-\alpha/2}^2}}{2B}\cdot\ind\{B\le z_{1 -\alpha/2}\}.
  \label{eq:optlambda}
\end{equation}
Corresponding to (\ref{eq:optlambda}), the optimal length of the marginal CI (\ref{eq:cut1}) is $B + \sqrt{-B^2 + 2z_{1-\alpha/2}^2}$ when $B\le z_{1-\alpha/2}$ and $2z_{1-\alpha/2}$ otherwise. This is identical to the partial conditioning IM solution derived in \citet{YangEtAl2023}.

\section{An Improved Regularization Solution}
\label{s:second}

We proceed to demonstrate that a slight change made to the test statistic results in more efficient marginal inference for $\theta_2$ in the two-normal-means problem. We also establish that the improvement is uniform across all $B > 0$ and $\alpha\in (0, 1)$.

\subsection{Derivation of the Confidence Interval}
\label{ss:ci}
We consider again the regularized ML estimator, $\hat\theta_2(y; \lambda)$. This time, we defined the following Wald-type statistic without centering:
\begin{equation}
  T_2(Y, \theta; \lambda) = \frac{\left[\lambda (Y_1 - \theta_2) + (1 + \lambda)(Y_2 - \theta_2)\right]^2}{\lambda^2 + (1 + \lambda)^2}.
  \label{eq:wald2}
\end{equation}
Compared to (\ref{eq:wald1}), (\ref{eq:wald2}) does not contain the additional centering term $\lambda(\theta_1 - \theta_2)$ within the bracket in the numerator. By (\ref{eq:diff}) under $\pr_{Y|\theta}$, $T_2(Y, \theta; \lambda)$ follows a noncentral chi-square distribution with degree of freedom 1 and noncentrality parameter $\lambda^2(\theta_1 - \theta_2)^2[\lambda^2 + (1 + \lambda)^2]^{-1}$. The joint possibility contour for $\theta$ corresponding to \ref{eq:wald2} can then be expressed as
\begin{equation}
  \begin{aligned}
    \pi_{y, 2}(\theta; \lambda) &=  \pr_{Y|\theta}\{T_2(Y, \theta; \lambda)\ge T_2(y, \theta; \lambda)\}\\
    &= 1 - F\left( \frac{\left[\lambda (y_1 - \theta_2) + (1 + \lambda)(y_2 - \theta_2)\right]^2}{\lambda^2 + (1 + \lambda)^2} ; 1, \frac{\lambda^2(\theta_1 - \theta_2)^2}{\lambda^2 + (1 + \lambda)^2}\right).
  \end{aligned}
  \label{eq:waldcontour2}
\end{equation}
Note that (\ref{eq:waldcontour2}) is symmetric around $[\lambda y_1 + (1 + \lambda)y_2]/(1 + 2\lambda)$. Moreover, it increases as $(\theta_1 - \theta_2)^2$ increases due to stochastic monotonicity of the non-central chi-square distribution with respect to its non-centrality parameter. Therefore, for each fixed $\theta_2$, the supremum of (\ref{eq:waldcontour2}) over $\theta_1 \in[\theta_2 - B, \theta_2 + B]$ is attained at $(\theta_1 - \theta_2)^2 = B^2$. The resulting supremum serves as a valid marginal possibility contour for $\theta_2$ and has the following expression:
\begin{equation}
  \begin{aligned}
    \varpi_{y, 2}(\theta_2; \lambda) &=  \sup_{\theta_1\in[\theta_2 - B, \theta_2 + B]}\pi_{y, 2}(\theta; \lambda)\\
    &=  1 - F\left( \frac{\left[\lambda (y_1 - \theta_2) + (1 + \lambda)(y_2 - \theta_2)\right]^2}{\lambda^2 + (1 + \lambda)^2}; 1, \frac{\lambda^2B^2}{\lambda^2 + (1 + \lambda)^2}\right).
  \end{aligned}
  \label{eq:mwaldcontour2}
\end{equation}

Let $g(\lambda, B) = \lambda^2B^2[\lambda^2 + (1 + \lambda)^2]^{-1}$ be the noncentrality parameter in (\ref{eq:mwaldcontour2}) and $Q_{1-\alpha}(\gamma)$ the $(1 - \alpha)$th quantile of the $\chi^2(1, \gamma)$ distribution. Solving $\theta_2$ from $\varpi_{y, 2}(\theta_2; \lambda) = \alpha$ yields the following upper $\alpha$-cut of (\ref{eq:mwaldcontour2}):
\begin{equation}
  \begin{aligned}
    C_2(\alpha; y) =  \bigg[ &
      \frac{\lambda y_1 + (1 + \lambda)y_2 - \sqrt{Q_{1 - \alpha}\{g(\lambda, B)\}[\lambda^2 + (1 + \lambda)^2]}}{1 + 2\lambda},\\
      &\frac{\lambda y_1 + (1 + \lambda)y_2 + \sqrt{Q_{1 - \alpha}\{g(\lambda, B)\}[\lambda^2 + (1 + \lambda)^2]}}{1 + 2\lambda}
  \bigg].
  \end{aligned}
  \label{eq:cut2}
\end{equation}
The length of (\ref{eq:cut2}) is 
\begin{equation}
  L_2(\lambda; \alpha, B) = \frac{2}{1 + 2\lambda}\sqrt{Q_{1 - \alpha}\{g(\lambda, B)\}[\lambda^2 + (1 + \lambda)^2]}.
  \label{eq:len2}
\end{equation}
Similar to the first regularized solution (Section \ref{s:first}), the expression (\ref{eq:len2}) does not depend on $y$, which allows us to find an optimal penalty weight that minimizes the length $L_2(\lambda; \alpha, B)$. However, the dependency of the length on $\lambda$ is analytically intractable due to the involvement of the non-central chi-square quantile. Next, we establish two results about $L_2(\lambda; \alpha, B)$. The first result states that the length as a function of the penalty weight $\lambda$ has a minimum, so that finding the optimal weight numerically is feasible. The second result concerns the comparison of $L_2(\lambda; \alpha, B)$ and $L_1(\lambda; \alpha, B)$ for any fixed triplet $\lambda$, $\alpha$, and $B$, from which we conclude that our new regularization solution dominates the partial conditioning solution of \cite{YangEtAl2023} in terms of efficiency.

\subsection{Analytical Results}
\label{ss:result}
In Figure \ref{fig:len}A, we plot the length functions for $\alpha = 0.05, 0.1$, and 0.2, respectively, while fixing $y = (1, 0.5)$ and $B = 1$. It appears from the graph that the minimum of the length function is always attainable. We justify this observation in Proposition \ref{prop:min}, which makes numerical search for the optimal penalty weight feasible.
\begin{figure}[!t]
  \centering
  \includegraphics[width=\textwidth]{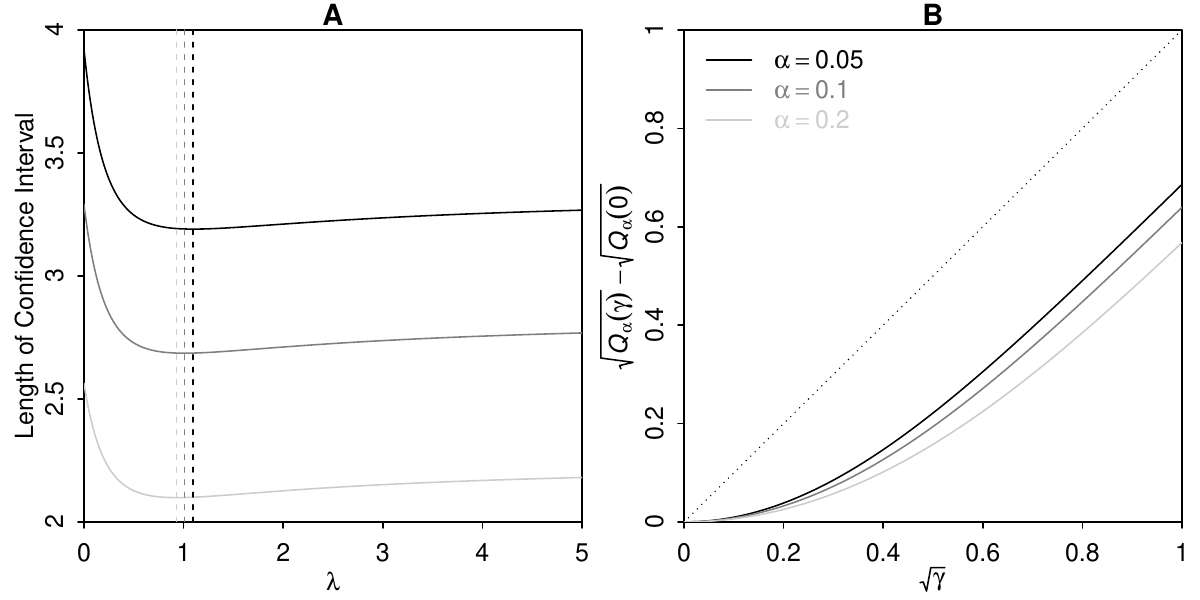}
  \caption{Illustrations for the length function $L_2(\lambda; \alpha, B)$. Panel A: Length functions for $\alpha = 0.05$ (black), 0.1 (dark gray), and 0.2 (light gray), fixing observed data $y = (1, 0.5)$ and bound $B = 1$. The minimums, found approximately by numerical search, are indicated by vertical dashed lines with matching colors. Panel B: $\sqrt{Q_\alpha(\gamma)} - \sqrt{Q_\alpha(0)}$ plotted against $\sqrt{\gamma}$, where $Q_\alpha(\gamma)$ is the $\alpha$th quantile of $\chi^2(1, \gamma)$ and $\gamma\ge 0$ is the noncentrality parameter.}
  \label{fig:len}
\end{figure}
\begin{proposition}
  For any given $\alpha\in(0, 1)$ and $B > 0$, the length $L_2(\lambda; \alpha, B)$ is minimized at some $\lambda\in(0, \infty)$.
  \label{prop:min}
\end{proposition}
\begin{proof}
  Differentiating $L_2(\lambda; \alpha, B)$ with respect to $\lambda$ yields
  \begin{equation}
    \begin{aligned}
      & L_2'(\lambda; \alpha, B) = (1 + 2\lambda)^{-2} \left[\lambda^2 + (1 + \lambda)^2\right]^{-3/2} Q_{1-\alpha}(g(\lambda, B))^{-1/2}\\
      &\qquad\cdot \left\{2 \lambda(1+\lambda)(1+2\lambda) Q_{1-\alpha}'(g(\lambda, B) )-2 \left[\lambda^2 + (1 + \lambda)^2\right] Q_{1-\alpha}(g(\lambda, B))\right\}.
    \end{aligned}
    \label{eq:deriv}
  \end{equation}
  On the one hand, $g(0, B) = 0$ and thus $L_2'(0; \alpha, B) = -2\sqrt{Q_{1-\alpha}(0)} < 0$. On the other hand, $g(\lambda, B)\to B^2/2$ as $\lambda\to\infty$, and both $Q_{1-\alpha}(\gamma)$ and $Q_{1-\alpha}'(\gamma)$ are positive for all $\gamma\ge 0$. Then there exists $\bar\lambda\in(0, \infty)$ such that (\ref{eq:deriv}) is positive for all $\lambda \ge \bar\lambda$. This is because $2\lambda(1+\lambda)(1 + 2\lambda)$, the polynomial multiplied to $Q'_{1-\alpha}(g(\lambda, B))$, is cubic in $\lambda$ while $2[\lambda^2 + (1 + \lambda)^2]$, the polynomial multiplied to $Q_{1-\alpha}(g(\lambda, B))$, is only quadratic. Taken together, (\ref{eq:deriv}) is negative at $\lambda = 0$ and positive at $\lambda = \bar\lambda$. By the continuity of $L_2'(\lambda; \alpha, B)$, its minimum is attained somewhere between 0 and $\bar\lambda$.  
\end{proof}

In the next proposition, we show that the new regularization-based CI is no wider than the partial conditioning CI for any fixed $\lambda, \alpha$, and $B$. The proof of the proposition relies on a technical lemma (Lemma \ref{lem:noncent}) about the non-central chi-square quantile. A graphical illustration of the Lemma \ref{lem:noncent} is presented in Figure \ref{fig:len}B.

\begin{lemma}
  For any given $\alpha\in(0, 1)$, $\sqrt{Q_\alpha(\gamma)} - z_\alpha\le\sqrt{\gamma}$ for all $\gamma\ge 0$. Moreover, the inequality is strict for all $\gamma > 0$.
  \label{lem:noncent}
\end{lemma}

\begin{proof}
  Let $\mu = \sqrt{\gamma}$ and $h(\mu) = \sqrt{Q_\alpha(\mu^2)}$. Because $h(0) - z_\alpha = \sqrt{0} = 0$, it suffices to show that $h(\mu)$ is strictly increasing in $\mu$, or equivalently $h'(\mu) > 0$, for all $\mu > 0$. As $Z\sim\sN(\mu, 1)$ implies $Z^2\sim\chi^2(1,\mu^2)$, $h(\mu)$ satisfies the following equation:
  \begin{equation}
    \Phi(h(\mu) - \mu) - \Phi(-h(\mu) - \mu) = \alpha,
    \label{eq:qchisq}
  \end{equation}
  in which $\Phi(\cdot)$ is the CDF of $\sN(0, 1)$. Differentiating both sides of (\ref{eq:qchisq}) with respect to $\mu$ and rearranging yields
  \begin{equation}
    h'(\mu) = \frac{\phi(h(\mu) - \mu) - \phi(- h(\mu) - \mu)}{\phi(h(\mu) - \mu) + \phi(- h(\mu) - \mu)} = \frac{\phi(|h(\mu) - \mu|) - \phi(|h(\mu) + \mu|)}{\phi(|h(\mu) - \mu|) + \phi(|h(\mu) + \mu|)},
    \label{eq:dqchisq}
  \end{equation}
  in which $\phi(\cdot)$ is the density of $\sN(0, 1)$. The second equality in (\ref{eq:dqchisq}) holds because $\phi(x) = \phi(-x) = \phi(|x|)$. For the reasons that $\phi(|x|)$ is strictly decreasing in $|x|$ and that $(h + \mu)^2 - (h - \mu)^2 = 4h\mu > 0$ for all $h, \mu > 0$, the numerator on the right-hand side of (\ref{eq:dqchisq}), and consequently $h'(\mu)$, is positive.
\end{proof}

\begin{proposition}
  For any given $\lambda\ge 0$, $\alpha\in(0, 1)$ and $B > 0$, $L_2(\lambda; \alpha, B)\le L_1(\lambda; \alpha, B)$.
  \label{prop:comp}
\end{proposition}
\begin{proof}
  Taking the difference between (\ref{eq:len1}) and (\ref{eq:len2}) yields
  \begin{equation}
    \begin{aligned}
    &L_1(\lambda; \alpha, B) - L_2(\lambda; \alpha, B)\\
    =&  \frac{2\sqrt{\lambda^2 + (1 + \lambda)^2}}{1 + 2\lambda}\left[ \sqrt{g(\lambda, B)} + z_{1 - \alpha/2} - \sqrt{Q_{1 - \alpha}(g(\lambda, B))} \right].
    \end{aligned}
    \label{eq:lendiff}
  \end{equation}
  The result follows immediately from Lemma \ref{lem:noncent}.
\end{proof}

\begin{figure}[!t]
  \centering
  \includegraphics[width=0.95\textwidth]{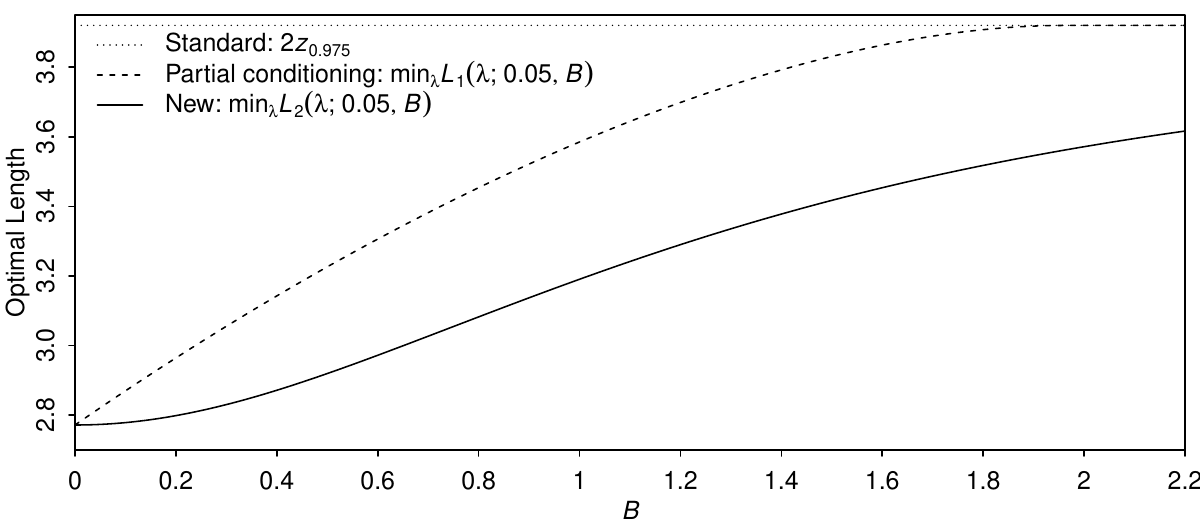}
  \caption{Comparison of confidence interval lengths.}
  \label{fig:comp}
\end{figure}

To visualize the comparison, we plot the respective optimal length functions, $\min_{\lambda\ge 0}L_1(\lambda; \alpha, B)$  and $\min_{\lambda\ge 0}L_2(\lambda; \alpha, B)$ against $B$ on the same graph (Figure \ref{fig:comp}). Again, the observed data $y = (1, 0.5)$. $\alpha$ is fixed at 0.05 while $B$ varies between 0 and 2.2. It is observed that the optimally tuned partial conditioning CI is much shorter than the standard CI for all $B < z_{0.975}\approx 1.96$. When $B\ge z_{0.975}$, the partial conditioning and standard CIs are identical. Meanwhile, our new solution based on regularization yields further efficiency gain compared to the partial conditioning solution across all values of $B$.

\section{Discussion}
\label{s:disc}

This work establishes regularization as an effective strategy for improving efficiency while preserving finite-sample validity of statistical inference. Focusing on the two-normal-means problem, we not only reproduce the partial conditioning IM solution proposed in \citet{YangEtAl2023} using a simple, regularization-based argument, but also manage to derive a more efficient regularized-based IM solution using a slightly different test statistic. We establish analytically and illustarte numerically that, compared to the partial conditioning IM, our proposal always results in narrower CIs.

An immediate question is whether the new regularized-based IM can be extended to the more general problem tackled by \citet{YangEtAl2023}: the many-normal-means problem with H\"older constraints. We conjecture that the general partial conditioning solution therein can be equivalently constructed through regularized ML estimation. The only difference is that we may need separate penalty terms for squared differences between adjacent observations. It is also contemplated that a suitable modification to the Wald-type test statistic can secure a similar efficiency gain.

For future work, we will continue to explore the integration of regularization with possibilistic IM in order to achieve valid and efficient inference in finite samples. The key insight is that a regularized estimator generates a family of valid inferential procedures indexed by penalty weights. It is then possible to select penalty weights to achieve the optimal efficiency. The major challenge in a more general context is that the measure of efficiency (e.g., the length of a CI) may depend on both data and parameters. Marginalization over the data and/or parameter space calls for more involved numerical search routines. In addition, due to the non-uniqueness of size measures in multidimensional parameter spaces, how to measure efficiency of general confidence regions remains to be an open question. Lastly, even further improvement in efficiency may be possible by first marginalizing the test statistic via the supremum operation; as in the profile-based marginal IM versus the extension-based marginal IM discussed in \cite{MartinWilliams2025}.

\setcounter{secnumdepth}{0}
\bibliographystyle{apalike}
\bibliography{NormalMeansReg}

\end{document}